\documentclass[preprint,12pt]{elsarticle}

\usepackage{amssymb}
\usepackage{amsthm}
\usepackage{subcaption}

\newtheorem{lemma}{Lemma}
\newtheorem{definition}{Definition}
\usepackage{amsmath} 
\usepackage[boxruled, linesnumbered,vlined]{algorithm2e}
\usepackage{color}
\usepackage{url}
\usepackage{booktabs}
\usepackage{longtable}
\usepackage{multicol}
\usepackage{tikz}
\usetikzlibrary{positioning, arrows.meta, shapes}
\SetKwRepeat{Do}{do}{while}
\usepackage{xltabular} 
\usepackage{booktabs}  
\usepackage{ltablex} 
\keepXColumns         
\usepackage{multirow}
\usepackage{graphicx}

\journal{Computers \& Operations Research}

\begin{document}

\begin{frontmatter}



\title{A Near-Real-Time Reduction-Based Algorithm for Coloring Massive Graphs}


\author[UESTC]{Chenghao Zhu}
\author[UESTC]{Yi Zhou\corref{cor}}
\cortext[cor]{Corresponding author} \ead{zhou.yi@uestc.edu.cn}
\affiliation[UESTC]{organization={University of Electronic Science and Technology 
 of China},
            city={Chengdu},
            country={China}}
\begin{abstract}
The graph coloring problem is a classical combinatorial optimization problem with important applications such as register allocation and task scheduling, and it has been extensively studied for decades. However, near-real-time algorithms that can deliver high-quality solutions for very large real-world graphs within a strict time frame remain relatively underexplored. In this paper, we try to bridge this gap by systematically investigating reduction rules that shrink the problem size while preserving optimality. For the first time, domination reduction, complement crown reduction, and independent set reduction are applied to large-scale instances. Building on these techniques, we propose RECOL, a reduction-based algorithm that alternates between fast estimation of lower and upper bounds, graph reductions, and heuristic coloring. We evaluate RECOL on a wide range of benchmark datasets, including SNAP, the Network Repository, DIMACS10, and DIMACS2. Experimental results show that RECOL consistently outperforms state-of-the-art algorithms on very large sparse graphs within one minute. Additional experiments further highlight the pivotal role of reduction techniques in achieving this performance.
\end{abstract}



\begin{keyword}
graph coloring \sep large graphs \sep reduction \sep near-real-time algorithm
\end{keyword}

\end{frontmatter}


\label{}
\section{Introduction}

Given an undirected graph $G$, the graph coloring problem seeks an assignment of colors to vertices such that adjacent vertices receive different colors, while minimizing the total number of colors used. The minimum number of colors required for a graph is called its \textit{chromatic number}, and the problem is therefore also known as the \textit{chromatic number problem}. Graph coloring is a fundamental problem in combinatorial optimization, operations research, and graph theory. The decision version was included in Karp’s seminal list of 21 NP-complete problems in 1972, and even the 3-coloring problem is NP-complete. Moreover, it is NP-hard to approximate the chromatic number within a factor of $n^{1-\epsilon}$ for any $\epsilon > 0$ \cite{zuckerman2006linear}. From a parameterized complexity perspective, the problem is W[1]-hard with respect to the number of colors $k$, suggesting that it is unlikely to admit an FPT algorithm running in time $f(k)n^{O(1)}$ \cite{downey2012parameterized}.

The graph coloring problem has a broad range of industrial applications, including register allocation, course timetabling, train platform assignment, frequency assignment, and other scheduling or resource allocation tasks. A common modeling approach represents entities as vertices, conflicts as edges, and resources or time slots as colors. For instance, in frequency assignment for wireless networks, nearby transmitters must be allocated different frequencies to avoid interference, and the chromatic number corresponds to the minimum spectrum required. As a matter of fact, graph coloring has long been regarded as a classical textbook problem.

Due to the great importance of the graph coloring problem, there exists a rich number of empirical solution approaches despite its theoretical hardness. 
These approaches can be divided into exact algorithms and heuristic (or inexact) algorithms. 
A more detailed review of these algorithms is deferred to the next section. 
In general, the 2nd DIMACS challenge, held in the 1990s, initiated the competition for finding chromatic numbers for a set of 87 benchmark graphs~\cite{dimacs2}. 
Since then, a great number of empirical algorithms have been proposed and used these graphs as benchmark sets.
To date, most of these DIMACS graphs have already been closed (most of them are proven optimal by exact algorithms). 
There are actually only a few instances for which existing algorithms cannot prove the optimal value.
However, to solve all these instances, these dedicated exact or heuristic algorithms need to run for days or even weeks to pursue the optimal answer. 
It is clear that this configuration is unacceptable in real-world scenarios, especially when real-time response is much more important than optimality.
For example, in the platform of frequency assignment or in compilers that ask for register allocation, the response should be in real-time, e.g., no more than one minute.

On the other hand, graphs from 2nd DIMACS challenges are typically dense, i.e., with an average edge density\footnote{The edge density is computed by $\frac{|E|}{\binom{|V|}{|2|}}\times 100\%$} of $55.91\%$, artificially generated, i.e., regular or large known cliques embedded in several graphs, and with no more than 4000 vertices. 
This is almost opposite to the real-world graphs, which are sparse, have no fixed structures, and are arbitrarily structured.
Particularly, the rapid development of data science in recent decades has brought many very large-scale, sparse real-world graphs. A large part of these graphs has at least 1 million vertices, as seen in datasets like SNAP~\cite{snapnets} and the Network Repository~\cite {networkrepository}. 
In fact, a later competition hosted by DIMACS, the 10th DIMACS challenge in 2010, introduced a larger and somehow modern collection of graphs, mixing extremely large and sparse real-world networks with artificially generated ones (e.g., Erdős–Rényi graphs) \cite{bader2013graph,bader2018benchmarking}.
However, according to our literature review, only a few existing approaches primarily focus  on obtaining reasonable coloring solutions within practical time constraints in these graphs \cite{cai2015balance,lin2017reduction}.
Hence, to date, it is believed there is still a strong need for graph coloring algorithms that produce near-optimal solutions for very large, sparse graphs in a short time, e.g., within one minute. 


\subsection{Main Contribution}

The main contributions of this paper are as follows:
\begin{itemize}
    \item We conduct a systematic study of data reduction rules for the graph coloring problem. A reduction rule transforms a graph into a smaller instance while preserving optimality. Beyond the commonly used degree-based rules, we introduce and implement domination reduction, limited complement crown reduction, and independent set reduction. These rules significantly decrease instance size and thereby facilitate efficient coloring of very large graphs.
    \item We develop a new heuristic algorithm, \textit{RECOL}, which embeds reduction rules into a cyclic framework that alternates between reductions, fast lower- and upper-bound estimation, and heuristic coloring. In contrast to existing approaches that apply reductions only as preprocessing, \textit{RECOL} repeatedly combines heuristic coloring (via independent set extraction) with reductions, enhancing their effectiveness and extending their applicability.
    \item We perform extensive experiments on diverse benchmark datasets, including the Network Repository, SNAP, DIMACS10, and DIMACS2. Under a strict one-minute time limit, \textit{RECOL} outperforms state-of-the-art algorithms on most benchmark sets. The only exception arises in extremely dense DIMACS2 graphs, where reductions are inherently limited. Additional experiments and ablation studies further demonstrate the critical role of reductions in achieving strong performance.
\end{itemize}

All codes and appended data are publicly available at \url{https://github.com/axs7385/RECOL}.






\subsection{Paper Organization}
The remainder of this paper is organized as follows. Section 2 reviews related work. Section 3 introduces the notation and key concepts used in the problem formulation. Section 4 presents the main reduction rules for graph simplification, including degree-based, domination, crown, and independent set reductions. Section 5 integrates these reductions with fast lower- and upper-bound estimation procedures into a complete near-real-time algorithm. Section 6 reports the experimental evaluation, including benchmark comparisons and an analysis of the impact of reductions. Finally, Section 7 concludes the paper and discusses directions for future research.

\section{Related Work}

Graph coloring has been a central focus of algorithmic research for decades, with solution approaches broadly divided into three primary categories: \textit{short-term heuristic methods}, \textit{long-term heuristic methods}, and \textit{exact methods}.  
Each of these categories is suitable for different types of graphs and application scenarios.

\paragraph{Short-term heuristic methods} 
Short-term algorithms aim to obtain a feasible coloring within a very limited time, typically under one minute. Most of these fast heuristics are based on greedy principles. Within this class, one family of approaches traverses vertices in a chosen order, assigning feasible colors sequentially, while another repeatedly extracts independent sets, assigning each set a single color.

In the vertex-ordering family, a wide variety of strategies have been proposed. Early work focused on static orderings, such as Largest First and eigenvalue-based schemes~\cite{hoffman1970eigenvalues,chow1990priority}. Later, dynamic orderings gained prominence, with the saturation degree heuristic becoming the most widely adopted. The most influential representative is \emph{DSatur}, which iteratively selects the uncolored vertex with the largest number of distinct neighbor colors~\cite{brelaz1979new}. DSatur remains a classic greedy heuristic and serves as a building block for many subsequent algorithms, both heuristic and exact. For example, \emph{FastColor}~\cite{lin2017reduction} alternates between reduction techniques and DSatur-based coloring, although its reductions are restricted to degree-based rules.

More recently, learning-based orderings have been introduced, where machine learning models such as graph neural networks are trained to predict effective vertex orderings~\cite{ijaz2022solving,langedal2025graph}. These methods extend the traditional greedy framework with data-driven guidance, often producing higher-quality orderings that adapt to diverse graph distributions. However, their benefits come with additional computational overhead for model training and inference, which may limit practicality on very large or time-critical instances.

The independent-set-based strategy is represented most prominently by the \emph{Recursive Largest First (RLF)} algorithm~\cite{leighton1979graph}, which iteratively identifies large independent sets, assigns them a new color, removes the corresponding vertices, and continues on the residual graph. Like DSatur, RLF has become a cornerstone for many heuristic and hybrid approaches, and its principle of independent-set construction continues to inspire modern adaptations~\cite{palubeckis2008recursive}.

\paragraph{Long-term heuristic methods}  
 Long-term heuristics are designed to achieve high-quality solutions, with runtimes that may extend to days or weeks. Unlike short-term heuristics, which emphasize speed, long-term methods prioritize solution quality and robustness, and often deliver state-of-the-art results on large benchmarks. Their defining characteristic is the systematic exploration of the solution space, typically through long-term search strategies that combine diverse neighborhoods, sophisticated diversification mechanisms, and intensive refinement phases.

Among these, local search methods such as \emph{Tabu Search} have been particularly influential~\cite{zufferey2008graph,hao2012improving}. Tabu Search employs a tabu list to forbid recently applied moves, preventing cycling, and many variants use large-neighborhood or swap-based moves to escape local optima. Beyond Tabu Search, other metaheuristics have also been widely applied, including evolutionary algorithms~\cite{eiben1998graph,galinier1999hybrid,goudet2021population}, simulated annealing~\cite{pal2012comparative,kole2022solving}, and configuration checking~\cite{wang2020reduction}. A survey by Galinier and Hertz provides a comprehensive overview of algorithms developed before 2006~\cite{galinier2006survey}.

More recent advances include \emph{GC-SLIM}, a Tabu Search-based algorithm that repeatedly extracts small induced subgraphs, solves them via SAT as list-coloring instances, and propagates results through chaining~\cite{schidler2023sat}. It has proven highly effective for large-scale instances. Very recently, \emph{LS-WGCP}, a local search algorithm for weighted graph coloring, was proposed~\cite{pan2025towards}. Unlike traditional tabu-based methods, it incorporates a configuration checking strategy to avoid cycling.

\paragraph{Exact Methods}
Exact algorithms aim to compute the chromatic number or certify optimality through exhaustive search combined with advanced pruning techniques. Common approaches include branch-and-bound~\cite{furini2017improved,san2012new,malaguti2011exact,morrison2014wide}, branch-and-price~\cite{malaguti2015branch}, SAT-based methods~\cite{ijcai2019p856,sewell1996improved}, constraint programming~\cite{10.1007/978-3-030-30048-7_13,gualandi2012exact}, and integer linear programming~\cite{burke2010supernodal,margot2009symmetry,mendez2008cutting}, as well as hybrid techniques that combine these paradigms. Exact methods are effective for small to medium-sized instances and can exploit special structural properties, even in dense graphs. However, their runtime and memory demands become prohibitive for very large sparse graphs arising in real-world applications.

\section{Preliminaries}

In this section, we introduce the basic notation used throughout the paper.

Let $G = (V, E)$ be a simple undirected graph, where $V = \{v_1, v_2, \ldots, v_n\}$ is the vertex set and $E = \{e_1, e_2, \ldots, e_m\}$ is the edge set. The complement graph of $G$ is denoted by $\overline{G} = (V, \overline{E})$, where $\overline{E} = \{\{u, v\} \mid \{u, v\} \notin E, u \neq v\}$. In other words, $\overline{G}$ has the same vertex set as $G$, but its edge set consists of all edges absent from $G$.

For a vertex $v \in V$, the \textit{open neighborhood} is defined as $N_G(v) = \{u \mid \{u, v\} \in E\}$, and the \textit{closed neighborhood} is $N_G[v] = N_G(v) \cup \{v\}$. For a subset $S \subseteq V$, the open neighborhood is $N_G(S) = \bigcup_{u \in S} N_G(u)$, and the closed neighborhood is $N_G[S] = \bigcup_{u \in S} N_G[u]$. When the context is clear, we abbreviate $N(u) = N_G(u)$, $\overline{N}(u) = N_{\overline{G}}(u)$, $N(S) = N_G(S)$, and $\overline{N}(S) = N_{\overline{G}}(S)$. The same notation applies to closed neighborhoods.

A subset $I \subseteq V$ is called an \textit{independent set} if no two vertices in $I$ are adjacent, i.e., $\{u, v\} \notin E$ for all distinct $u, v \in I$. Conversely, a subset $C \subseteq V$ is a \textit{clique} if every two vertices in $C$ are adjacent, i.e., $\{u, v\} \in E$ for all distinct $u, v \in C$.

A feasible coloring of $G = (V,E)$ is a partition of $V$ into pairwise disjoint independent sets $\mathcal{C}$ such that $\bigcup_{I \in \mathcal{C}} I = V$. We refer to such a partition as an \textit{independent set cover}. The \textit{chromatic number} of $G$, denoted by $\chi(G)$, is the minimum cardinality of an independent set cover of $G$. We denote the set of all minimum independent set covers of $G$ by $ic(G)$.


\section{Practical Reduction Rules}
We first discuss some practical reduction rules for the graph coloring problem. 
A reduction rule for the graph coloring is a mapping that, given a graph $G$, checks a specific condition and, if satisfied, transforms $G$ into another instance $G'$ such that:
\begin{itemize}
    \item (Equivalence) The optimal solution to $G'$ can be extended in polynomial time to an optimal solution of $G$ without missing the optimal objective value.
    \item (Size Decrease) The instance $G'$ is smaller than $G$ under some natural measure (e.g., fewer vertices, fewer edges).
\end{itemize}

\subsection{Degree Reduction}
Generally, vertices with smaller degrees are easier to color. 
Thus, we can color these vertices after coloring the other vertices.

\begin{lemma}
Let $G=(V,E)$ be a graph and $u\in V$. If $|N(u)|<\chi (G[V \setminus \{u\}])$, then $\chi(G)=\chi(G[V\setminus \{u\}])$. Moreover, for every $\mathcal{C}\in ic(G[V\setminus \{u\}])$, there exists an independent set $I' \in \mathcal{C}$ such that $I' \cup \{u\}$ remains an independent set in $G$,  and we can construct a valid independent set cover 
\[
\mathcal{C}'=(\mathcal{C}\setminus \{I'\})\cup\{I'\cup\{u\}\} \in \mathcal{C}'\in ic(G).
\]
\end{lemma}

\begin{proof}

Let $ G' = G[V\setminus \{u\}] $ and $\mathcal{C}\in ic(G')$. Given that $ G' $ is a subgraph of $ G $, we have $ \chi(G') \leq \chi(G) $ and $ N(u) \subseteq G' $.

Since $|\mathcal{C}|=\chi(G')>|N(u)|$, there exists at least one independent set $I'\in \mathcal{C}$ where $\forall v\in I',\{u,v\} \notin E$. Thus, $I' \cup \{u\}$ is still an independent set. Furthermore, we can prove that $\mathcal{C}'=(\mathcal{C}\setminus \{I'\})\cup\{I'\cup\{u\}\}$ is an independent set cover of $G$, because $\{I'\cup\{u\}\}$ is an independent set
and $\mathcal{C}'$ includes the vertex of $(V(G') \setminus I')\cup I'\cup\{u\} = V(G')\cup \{u\} = V(G)$.  

By the construction of $\mathcal{C}'$, we get $|\mathcal{C}'|=|(\mathcal{C}\setminus \{I'\})\cup\{I'\cup\{u\}\}|=(|\mathcal{C}|-1)+1=|\mathcal{C}|=\chi(G')$. In addition, because $\mathcal{C}'$ is an independent set cover of G, $|\mathcal{C}'| \geq \chi(G)$. Therefore, we find that $|\mathcal{C}'|=\chi(G')\leq \chi(G) \leq |\mathcal{C}'|$, so $\chi(G')= \chi(G)$ and $\mathcal{C}' \in ic(G)$.

To conclude, we have $ \chi(G) = \chi(G[V - \{u\}]) $.
\end{proof}

The lemma asks for the exact chromatic number. But it is difficult to determine this number. So we can apply some relaxation to the constraints to derive reduction rules that can be used practically.

\begin{lemma}
\label{degre_reduction}
Let $G=(V,E)$ be a graph, $u \in V$ a vertex, and let $lb$ be a lower bound on the chromatic number of $G$. If $|N(u)| < lb$, then
\[
\chi(G) = \max\{ lb, \chi(G[V \setminus \{u\}]) \}.  
\]
Moreover, for every independent set cover $\mathcal{C} \in ic(G[V \setminus \{u\}])$:
if there exists an independent set $I' \in \mathcal{C}$ such that $I' \cup \{u\}$ is an independent set in $G$, then
\[
\mathcal{C}' = (\mathcal{C} \setminus \{I'\}) \cup \{ I' \cup \{u\} \} \in ic(G);  
\]
otherwise,
\[
\mathcal{C}' = \mathcal{C} \cup \{\{u\}\} \in ic(G).  
\]
\end{lemma}


\begin{proof}
    
    Let $G' = G[V \setminus \{u\}]$, and let $lb$ be a valid lower bound on $\chi(G)$.  Since removing a single vertex can reduce the chromatic number by at most one, we have $\chi(G) - 1 \leq \chi(G') \leq \chi(G)$. Therefore, $\chi(G') \geq lb - 1$. 
    
    If $\chi(G') \geq lb$, according to Lemma 1, since  $|N(u)| < lb \leq \chi(G[V \setminus \{u\}])$, we obtain $\chi(G) = \chi(G[V - \{u\}])$. Otherwise, $\chi(G') = lb - 1$, hence $\chi(G) = lb = \chi(G') + 1$. In conclusion, $\chi(G) = \max(lb, \chi(G[V \setminus \{u\}]))$.

\end{proof}

Based on the lemma, it is easy to propose the \textit{degree reduction rule} -- When the degree of a vertex $u$ is less than the lower bound $lb$ of the chromatic number of the given graph, we can remove the vertex $u$ from the graph. 
 Degree reduction is one of the most important components in our algorithm.
 As discovered in further experiments, most vertices are removed during the degree reduction procedure. 
 Actually, degree reduction is the only reduction rule in the algorithm in \cite{lin2017reduction}.
 Nevertheless, to use degree reduction, it is necessary to find a sufficiently tight lower bound $lb$ beforehand. 
 


\subsection{Dominate Reduction}

\begin{definition}
In $G=(V,E)$, a vertex $v \in V$ dominates a vertex $u \in V$ if all neighbors of $u$ are also neighbors of $v$ (i.e. $N(u) \subseteq N(v)$). \end{definition}

Note that if we say $v$ dominates $u$, $v$ and $u$ must not be adjacent in the graph.

\begin{lemma}
\label{dominatereduction}
Let $G=(V,E)$ be a graph and let $u,v \in V$ be two vertices such that $v$ dominates $u$. Then,
$$
\chi(G) = \chi(G[V \setminus \{u\}]).  
$$
Moreover, for every independent set cover $\mathcal{C} \in ic(G[V \setminus \{u\}])$, there exists an independent set $I_v \in \mathcal{C}$ with $v \in I_v$ such that $I_v \cup \{u\}$ is also an independent set in $G$. In this case, a valid independent set cover of $G$ can be obtained as
$$
\mathcal{C}' = (\mathcal{C} \setminus \{I_v\}) \cup \{I_v \cup \{u\}\} \in ic(G).  
$$
\end{lemma}


\begin{proof}
Let $ G' = G[V\setminus \{u\}] $ and $\mathcal{C}\in ic(G')$. 
Since $ G' $ is a subgraph of $ G $, we have $ \chi(G') \leq \chi(G)$. 
Consider the vertex $v \in G'$ and let $I_v \in \mathcal{C}$ be the independent set that contains $v$. 
Because $N(u) \subseteq N(v)$ and $(u,v)\notin E$ , $I_v \cup \{u\}$ is still an independent set. First, $\mathcal{C}'=(\mathcal{C}\setminus \{I_v\})\cup\{I_v\cup\{u\}\}$ is an independent set cover of $G$, since $\cup\{I_v\cup\{u\}\}$ is an independent set and $\mathcal{C}'$ includes the vertex of $(V(G') \setminus I_v)\cup I_v\cup\{u\} = V(G')\cup \{u\} = V(G)$. 
Because $\mathcal{C}'$ is an independent set cover of G, $|\mathcal{C}'| \geq \chi(G)$. 
Second, by the construction of $\mathcal{C}'$, we get $|\mathcal{C}'|=|(\mathcal{C}\setminus I_v)\cup\{I_v\cup\{u\}\}|=(|\mathcal{C}|-1)+1=|\mathcal{C}|=\chi(G')$.
Since $\chi(G') \leq \chi(G)$ and $\chi(G) \leq |\mathcal{C}'| = \chi(G')$, we obtain $\chi(G) = \chi(G')$ and $\mathcal{C}' \in ic(G)$. 

\end{proof}

Based on the above lemma, we obtain the \textit{domination reduction} algorithm -- When the neighborhood of a vertex $u$ is a subset of the neighborhood of a vertex $v$, we remove vertex $u$ from the graph. 

Note that, if two vertices are mutually dominated, i.e., $N(u)=N(v)$, then we can only reduce either $u$ or $v$ arbitrarily, rather than both of them.
For example, consider a graph that is a chain of length three $G = (V, E)$, where $V = \{1, 2, 3\}$ and $E = \{\{1, 2\}, \{2, 3\}\}$. Clearly, $N(1) = N(3) = \{2\}$. This satisfies both $N(1) \subseteq N(3)$ and $N(3) \subseteq N(1)$. The chromatic number of the original graph is 2. If we reduce both vertices 1 and 3 simultaneously, the resulting chromatic number would incorrectly be 1. 

\subsection{Complement Crown Reduction}

\begin{definition}
Given a graph $G=(V,E)$ and a pair of disjoint vertex sets $(H,I)$ where $H\subseteq V$ and $I\subseteq V$ in $G=(V,E)$, if 
\begin{itemize}
    \item $I$ is an independent set,
    \item $N(I) = H$ and,
    \item there exists a bipartite matching $M$ between $H$ and $I$ of size $|H|$,    
\end{itemize}
   then $(H,I)$ is called a crown structure in $G$. Specifically, $H$ is called a head and $I$ a crown.
\end{definition}


\textit{Crown reduction} is a powerful reduction technique in problems such as vertex cover  \cite{downey2012parameterized} and clique cover \cite{strash2022effective}.
Because the graph coloring problem in $G$ is equivalent to the clique cover problem in its complement graph $\overline{G}$, it is natural to translate any reduction rule of clique cover to that of graph coloring in the complement graph. 
Therefore, we directly depict the (complement) crown reduction for graph coloring.
Let us first define the complement crown structure in a graph $G$.

\begin{definition}[Complement Crown]
Given a graph $G=(V,E)$ and a pair of disjoint vertex sets $(B,C)$ where $B\subseteq V$ and $C\subseteq V$,
\begin{itemize}
    \item if $C$ is a clique,
    \item $\overline{N}(C) = B$ (recall that $\overline{N}(C)$ is the union of neighbors of vertex in $C$ in the complement graph $\overline{G}$) and,
    \item  there exists a bipartite matching $M$ in $\overline{G}$ between $B$ and $C$ of size $|B|$
\end{itemize}
then we call the pair $(B,C)$ a complement crown structure in $G$. Specifically, $B$ is the complement head and $C$ is the complement crown.
\end{definition}

Note that $M$ is a family of non-adjacent vertex pairs in $G$ (so that it is a matching in $\overline{G}$). 
The complete crown reduction for graph coloring is based on this complement crown structure.


\begin{lemma}
Let $G = (V, E)$ be a graph and let $(B,C)$ be a complement crown in $G$. Suppose $M$ is a bipartite matching between $B$ and $C$ in the complement graph $\overline{G}$. Then,
$$
\chi(G) = \chi\!\left(G[V \setminus (B \cup C)]\right) + |C|.  
$$
Moreover, let $C' \subseteq C$ denote the subset of vertices in $C$ that are unmatched in $M$, i.e.,$ C' = \{\,u \in C \mid \nexists (u,p) \in M \,\}.$
For every independent set cover $\mathcal{C} \in ic(G[V \setminus (B \cup C)])$, a valid independent set cover of $G$ can be constructed as
$$
\mathcal{C}' = \mathcal{C} \cup \{\,\{u,v\} \mid (u,v) \in M \,\} \cup \{\{u\} \mid u \in C'\},  
$$
such that $\mathcal{C}' \in ic(G)$.
\end{lemma}

\begin{proof}
Let $ G' = G[V\setminus(B\cup C)] $ and $\mathcal{C}\in ic(G')$. 
Recall that $C' \subseteq C$ denotes the subset of vertices in $C$ that are unmatched in the matching $M$, i.e.,
$C' = \{ u \in C \mid \nexists (u,p) \in M \}$.
$C$ is a clique, so $\chi(G[C])=|C|$. 
Because $C$ is a subgraph of $B \cup C$, $\chi(G[B \cup C])\geq \chi(G[C]) = |C|$. 
We can construct $\mathcal{X} = M\cup \{\{u\}|u\in C'\}$, an independent set cover of $B\cup C$, and we have $\chi(G[B \cup C])\leq|\mathcal{X}|=|C|$.
Thus, $\chi(G[B \cup C]) = |C|$ and $\mathcal{X} \in ic(G[B \cup C])$.


\begin{itemize}
    \item $\chi(G) \leq \chi(G') + |C|$:
    
        $\mathcal{C} \cup \mathcal{X}$ is an independent set cover of G. Hence, $\chi(G)\leq |\mathcal{C} \cup \mathcal{X}|=|\mathcal{C}|+|\mathcal{X}|=\chi(G')+|C|$,
    \item $\chi(G) \geq \chi(G') + |C|$:

        Considering any independent set cover $Y$ of $G[V\setminus B]$, we have $\forall u\in C,\{u\}\in Y$. Therefore, $\chi(G[V\setminus B])=\chi(G')+|C|$. Because $G[V\setminus B]$ is a subgraph of $G$, we have $\chi(G)\geq \chi(G[V\setminus B])=\chi(G')+|C|$.
\end{itemize}

Thus, we conclude $\chi(G) = \chi(G[V \setminus (B \cup C)]) + |C|$.
\end{proof}

A crown structure can be identified in time $O(m\sqrt{n})$  in a graph with $n$ vertices and $m$ edges \cite{nemhauser1975vertex}.
If the input graph is sparse, i.e., assuming $m\approx O(n)$, then the edge number in the complement graph is about $O(n^2)$, the complexity of identifying the complement crown structure is $O(n^{2.5})$, which is almost impractical when  $n$ is large (e.g. $n$ is around 10,000).
Actually, even for the clique cover problem, the crown structure is not completely identified due to the prohibitive runtime \cite{strash2022effective}. 
Based on these observations,  we focus on efficiently checking for certain special cases of complement crown structure that can be easily recognized.
Specifically, we reduce the complement crown structure $(B, C)$ where $|C| \le 2$.  

\begin{lemma}
\label{crown1reduction}
Let $G = (V, E)$ be a graph and $u \in V$. If $|\overline{N}(u)| \leq 1$, then
$$
\chi(G) = \chi(G[N(u)]) + 1.  
$$
Moreover, for every independent set cover $\mathcal{C} \in ic(G[N(u)])$, the set $\overline{N}[u]$ is an independent set in $G$, and a valid independent set cover of $G$ can be obtained as
$$
\mathcal{C}' = \mathcal{C} \cup \{\overline{N}[u]\} \in ic(G).  
$$
\end{lemma}


\begin{lemma}
\label{crown2reduction}
Let $G = (V, E)$ be a graph and $\{u, v\} \in E$. If $\overline{N}(u) = \overline{N}(v) = \{x, y\}$, then
$$
\chi(G) = \chi(G[N(u) \setminus \{v\}]) + 2 = \chi(G[N(u)]) + 1.  
$$
Moreover, for every $\mathcal{C}\in ic(G[N(u)])$, there exist two independent sets $\{u,x\},\{v,y\}$, and a valid independent set cover of $G$ can be obtained as
$$
\mathcal{C}' = \mathcal{C} \cup \{\{u, x\}, \{v, y\}\} \in ic(G).  
$$
\end{lemma}


It is possible to discuss more special cases of the crown structure, such as cases where $|C|$ is at least $3$. 
However, our preliminary tests suggest that we can already achieve an acceptable reduction in size with these limited cases. It also becomes less practical to detect and reduce these complement crown structures where $|C|$ is larger. 




\subsection{Independent Set Reduction}


Given an independent set cover $\mathcal{C}'$, a vertex $v$ must be covered by a unique independent set $I_v\in \mathcal{C}'$. Suppose that $\overline{N}(v)$ forms an independent set. Then, it is clear that moving all vertices of $\overline{N}(v)$ to $I$ does not change the fact that $\mathcal{C}'$ is an independent set cover and does not increase its cardinality $|\mathcal{C}'|$. Based on this, we have the following independent reduction rule.

\begin{lemma}
\label{independentreduction}
Let $G=(V,E)$ be a graph and $u \in V$. If $\overline{N}[u]$ is an independent set, then
$$
\chi(G) = \chi(G[N(u)]) + 1.  
$$
Moreover, for every independent set cover $\mathcal{C} \in ic(G[N(u)])$, a valid independent set cover of $G$ can be obtained as
$$
\mathcal{C}' = \mathcal{C} \cup \{\overline{N}[u]\} \in ic(G).  
$$
\end{lemma}

\begin{proof}
Let $ G' = G[N(u)] $ and $\mathcal{C}\in ic(G')$. 
First, $\forall v\in N(u), \{u,v\} \in E$, so vertex $u$ cannot share one color with any node in $N(u)$, thus $\chi(G) \geq \chi(G[N[u]]) = chi(G')+1$.
Second, $\mathcal{C}'=\mathcal{C}\cup\{\overline{N}[u]\}$ is an independent set cover of $G$ because  $\overline{N}[u]$ is an independent set and $\mathcal{C}'$ includes the vertex of $N(u)\cup\overline{N}[u] = V$. We can find that $\chi(G)\leq|\mathcal{C}'|=\chi(G')+1$. 

Thus, we conclude that $\chi(G)=\chi(G')+1$, and therefore $\mathcal{C}\cup\{\overline{N}[u]\}\in ic(G)$.

\end{proof}

\section{A Reduction-Based Graph Coloring Algorithm - RECOL}

Based on the exploitation of reduction rules for the graph coloring problem, we propose RECOL, a reduction-based algorithm capable of solving instances with millions of vertices under one minute.
The main feature of the algorithm lies in the integration of graph reduction with fast lower- and upper-bound improvements within a unified framework.


\begin{algorithm}[H]
    \caption{The RECOL algorithm}
    \KwIn{The initial input graph $G=(V,E)$}
    \KwOut{The best found chromatic number of $G$}     
    RECOL($G=(V,E)$)\\
    \Begin{
    $ans\gets |V|, round\gets 0$\\
    \While{cutoff time is not met}
    {
        Reset $G$ as the original input graph
        $lb \gets 0,ub\gets|V|,usedcol\gets 0,round\gets round +1$\\
        \While{$lb < ub$ and cutoff time is not met}
        {
            $temp\_lb \gets \text{FindClique}(G,lb,usedcol)$\\
            \uIf{$round = 1$}{
                $temp\_ub\gets\text{Degenarcy}(G,ub,usedcol)$\\
            }
            \Else{
                $temp\_ub \gets \text{Dsatur}(G, ub, usedcol)$\\
            }
            $G'=(V',E'), usedcolor \gets$ the graph and used color numbers after reducing $G$ by reductions rules in Lemma \ref{degre_reduction}, \ref{dominatereduction}, \ref{crown1reduction}, \ref{crown2reduction} and \ref{independentreduction} \\
            \If {$|V'|<|V|$ or $temp\_lb>lb$ or $temp\_ub<ub$}
            {
                $G\gets G'$,$lb\gets temp\_lb, ub\gets temp\_ub$          
            }\Else{
                $I \gets $FindIndepedentSet$(G, round)$\\
                $usedcol \gets usedcol+1$ \\
                $V \gets V \setminus I$\\
            }
        }
        \If{$ub <ans$}
        {
            $ans\gets ub$\\
        }
    }
    \Return{$ans$}\\
    }
\end{algorithm}

The core strategy of \textit{RECOL} is to iteratively apply reduction rules that simplify the graph while simultaneously maintaining information about optimal solutions. In lines 4–20, the algorithm executes a loop that begins with the original graph and searches for a solution by intertwining reductions with continuous refinement of both lower and upper bounds.
The lower bound is estimated through a randomized maximal clique procedure (\texttt{FindClique()} at line~7), whereas the upper bound is obtained by combining a linear-time degeneracy-ordering algorithm (\texttt{Degeneracy()} at line~9) with a DSATUR-like algorithm (\texttt{Dsatur()} at line~11). The use of two distinct upper-bounding procedures is motivated by efficiency: in the first iteration, when the graph remains large and unreduced, the fast greedy \texttt{Degeneracy()} algorithm is applied to quickly provide a valid upper bound, while \texttt{Dsatur()} is employed when the graph has been reduced.

The reduction rules from Lemmas~\ref{degre_reduction}, \ref{dominatereduction}, \ref{crown1reduction}, \ref{crown2reduction}, and \ref{independentreduction} are applied to further shrink the graph $G$, yielding a reduced instance $G'$. Table~\ref{tab:rules} summarizes the role of each reduction rule.

\begin{table}[ht]
\centering
\caption{How reduction rules work in RECOL.}
\label{tab:rules}
\scriptsize
\begin{tabular}{lcc}
\toprule
 Lemma & Rule & Reduced Structure\\
\midrule
\ref{degre_reduction} & Degree reduction & Vertices $u$ such that $d(u)<lb$\\
\ref{dominatereduction} & Dominate reduction & Vertices $u$ such that $\exists v\in V, N(u)\subseteq N(v)$\\
\ref{crown1reduction},\ref{crown2reduction} & Complement crown reduction & Vertices $B\cup C$ where $(B,C)$ forms a complement crown\\
\ref{independentreduction} & Independent set reduction & Vertices $\overline{N}[u]$ where $\overline{N}[u]$ is an independent set\\
\bottomrule
\end{tabular}
\end{table}

If the graph size decreases, or if either the lower or upper bound is improved, we update the working graph together with the bounds $lb$ and $ub$. When no further reductions are possible, an independent set is extracted, assigned a new color, and removed from the graph. Clearly, the removal of an independent set may enable additional reductions in the next round.
This cycle of reduction, bounding, and partial coloring is repeated until the time limit is reached, at which point the algorithm returns the best chromatic number found. By integrating effective reductions with dynamic bound management, \textit{RECOL} scales efficiently to graphs with millions of vertices within practical time limits.

It is worth noting that detecting domination relationships requires $O(nm)$ time. For this reason, the domination rule is applied only to graphs with at most 200 vertices. Similarly, identifying an independent set of the form $\overline{N}[u]$ is computationally demanding ($O(n^3)$). Hence, the independent set reduction is restricted to vertices with at most 10 non-neighbors, since vertices of low degree are unlikely to form large independent sets together with their non-neighbors.

\subsection{Estimate the Lower Bound of $\chi(G)$}

To reduce more vertices, we rely on strong lower bounds of the chromatic number. Classical examples include the clique number $\omega(G)$, since all vertices of a clique must receive distinct colors; the bound $|V|/\alpha(G)$, where $\alpha(G)$ is the independence number, as a coloring partitions $V$ into independent sets; and relaxation-based bounds derived from linear or semidefinite programming. As this work targets very large graphs, it is essential to compute such bounds within a very short time. For this purpose, we adopt the heuristic maximum clique algorithm proposed in \cite{DBLP:conf/ijcai/LinCLS17}.

The procedure is outlined in Algorithm~\ref{algorithm:findclique}. Starting from $\epsilon |V|$ randomly selected seed vertices ($\epsilon = 0.01$ in our experiments, following \cite{DBLP:conf/ijcai/LinCLS17}), the algorithm expands each seed into a maximal clique. The size of the resulting clique, combined with the number of colors already used, produces an updated lower bound.

\begin{algorithm}[H]
    \caption{The clique heuristic to find a lower bound of $\chi(G)$ \label{algorithm:findclique}}
    \KwIn{Graph $G=(V,E)$, initial lower bound of chromatic number $lb$, extra color number $usedcol$}
    \KwOut{new lower bound of chromatic number $lb'$}
    FindClique($G=(V,E),lb,usedcol$)\\
    \Begin{
     $S \gets $ randomly select $\epsilon|V|$ vertices from $V$\;
     \For{$u \in S$}
     {
        $Clique \gets \{u\}$\;
        $Candidate \gets N(u)$\;
        \While{$Cand \neq \emptyset$ and $|Clique| + |Cand| + usedcol > lb$}
        {
            Randomly select some vertices in $Candidate$ to $Sample$
            $v \gets { \underset{u\in Sample}{{\arg\max}}}(|N(u)\cap Candidate|)$\;
            $Clique \gets Clique \cup \{v\}$\;
            $Candidate \gets Candidate \cap N(v)$\;
        }
        \If{$|Clique| + usedcol > lb$}
        {
            $lb \gets |Clique| + usedcol$\;
        }
     }
     return $lb$\;
     }
\end{algorithm}



\subsection{Estimate the Upper Bound of $\chi(G)$}

It is known that any valid heuristic coloring provides an upper bound for the chromatic number.
In the initial round of the whole coloring algorithm, we  want the upper-bounding algorithm to be scalable and efficient. Hence, we exploit a simple degeneracy-based  algorithm which runs in linear time $O(|E|)$ \cite{matula1983smallest}, as shown in Algorithm~\ref{algorithm:upperbound_estimate}. The degeneracy-based coloring algorithm repeatedly selects the vertex with the fewest uncolored neighbors and assigns it the smallest feasible color. Each time a new color is introduced, the color count is updated, and the procedure terminates early if the number of colors used reaches the current upper bound. The final result is returned as the updated upper bound.

When scalability is not the most important feature, we resort to one of the most famous heuristic algorithms, the Dsatur by Daniel Brélaz \cite{brelaz1979new}, for upper bound estimation.
The procedure is also shown in the Algorithm~\ref{algorithm:upperbound_estimate}.
In general, the algorithm iteratively assigns a color to the vertex with the highest \emph{saturation degree}, 
where the saturation degree of a vertex $u$ is defined as the number of distinct colors already used in its neighborhood. 
When multiple vertices share the maximum saturation degree, one of them is chosen at random to increase diversity in the coloring process.


\begin{algorithm}[H]
    \caption{The procedures of estimating an upper bound of $\chi(G)$.
    \label{algorithm:upperbound_estimate}}
    \KwIn{graph $G=(V,E)$, initial lower bound of chromatic number $ub$, extra color number $usedcol$}
    \KwOut{new lower bound of chromatic number $ub'$}
    Degenarcy($G=(V,E),lb,ub,excol$)\\
    \Begin{
    $Candidate \gets V$, $colCnt\gets 0$ \\
    \While{$Candidate \neq \emptyset$}
     {
        $v \gets $ one vertex in $Cand$ with the fewest uncolored neighbors\\
        $col[v] \gets $ the smallest valid color for vertex $v$\\        $Candidate \gets Candidate \setminus \{v\}$\;
        \If{$col[v]>colCnt$}
        {
            $colCnt\gets colCnt +1$\\
            \If {$usedcol+colCnt \geq ub$}
            {
                return $ub$\\
            }
        }
     }
     return $usedcol+colCnt$\\
     }
     
    Dsatur($G=(V,E),lb,ub,excol$)\\
    \Begin{
    $Candidate \gets V$, $colCnt\gets 0$ \\
    Initiate the  saturation of each vertex as $0$.\\
     \While{$Candidate \neq \emptyset$}
     {
        $v \gets $ one vertex in $Cand$ with maximum saturation\\
        $col[v] \gets $ a randomly valid color for vertex $v$\\
        $Candidate \gets Candidate \setminus \{v\}$\;
        \If{$col[v]>colCnt$}
        {
            $colCnt\gets colCnt +1$\\
            \If {$usedcol+colCnt \geq ub$}
            {
                return $ub$\\
            }
        }
        update saturation of $N(v)$\\
     }
     return $usedcol+colCnt$\\
     }
     
\end{algorithm}

\subsection{Independent Set Extraction}

In \textit{RECOL}, when no further reductions are possible, an independent set is extracted from the current graph. The extraction procedure is described in Algorithm~\ref{algorithm:independent extraction}. Starting from the candidate vertices, the algorithm repeatedly selects the maximum-degree vertex. With probability $p$ (the skip parameter, randomly chosen in the range $(0,0.25]$ in our experiments), the vertex is skipped; otherwise, it is added to the set and its neighbors are removed. By varying $p$ across rounds, different independent sets are generated, thereby enlarging the search space and improving coloring quality. The extracted set is then assigned a new color and removed from the graph, which allows further bound updates and reductions.

\begin{algorithm}[H]
    \caption{Find an Independent Set}
    \label{algorithm:independent extraction}
    \KwIn{Graph $G=(V,E)$, current round $round$}
    \KwOut{An independent set $I$}
    FindIndepedentSet$(G, round)$\\
    \Begin{
        $p \gets \frac{round \bmod 25}{100}$ \tcp{skip probability parameter}
        $I \gets \emptyset$,$Candidate \gets V$\\        
        \While{$Candidate \neq \emptyset$}
        {
            $v\gets \underset{u\in Candidate}{{\arg\max}}(|N(u)|)$\\
            $Candidate \gets Candidate \setminus\{u\}$\\
            \If{$rand(0,1)>p $}
            {
                $I \gets I \cup \{v\}$\;
                $Candidate \gets Candidate\setminus N(v)$
            }
        }
        \Return{$I$}\\
    }
\end{algorithm}

\section{Experiment}
In this section, we empirically evaluate the performance of the proposed algorithm.



\subsection{Experimental Setup}  


All experiments were conducted on a server equipped with an Intel(R) Xeon(R) Platinum 8360Y CPU @ 2.40GHz and 1 TB of RAM, running Ubuntu 22.04.3 LTS. 
The code was developed in C++ and compiled with GCC 11.4.0 using the -O2 optimization flag. 

We compare our algorithm, \textsc{RECOL}, against three state-of-the-art heuristic algorithms mentioned in the related work.
\begin{itemize}
    \item \emph{FastColor}\cite{lin2017reduction}: One of the fastest heuristics designed for large-scale sparse graphs, combining reduction techniques with DSatur-based coloring.
    \item  \emph{GC-SLIM}\cite{schidler2023sat}: A SAT-boosted tabu search approach that found the best-known results on challenging benchmark instances.
    \item \emph{LS-WGCP}\cite{pan2025towards}: A very recent heuristic algorithm for vertex-weighted graph coloring. When the weights of all vertices are equal, the problem is equal to graph coloring.
\end{itemize}
To the best of our knowledge, these algorithms constitute strong baselines covering both long-term and short-term heuristic algorithms.

\subsection{Datasets}

We evaluate the algorithms on four representative graph collections.

\begin{itemize}
    \item \textbf{Network Repository} \cite{networkrepository} (139 graphs): A broad collection of real-world graphs from social, biological, and technological domains. These graphs are generally sparse and irregular, with some medium-sized instances showing richer clustering and attribute information.  
    
    \item \textbf{SNAP} \cite{snapnets} (51 graphs): Large-scale networks with up to millions of nodes and edges, mainly from social, citation, and web domains. Compared with the Network Repository, these graphs focus more on size and sparsity.  
    \item \textbf{10th DIMACS Challenge}\cite{bader2013graph,bader2018benchmarking} (134 graphs): A collection combining real-world and synthetic instances, including Erdős–Rényi graphs, with sizes ranging from small to extremely large (up to millions of vertices). 
    
    \item \textbf{2nd DIMACS Challenge}\cite{dimacs2} (87 graphs): Artificially generated graphs with relatively small size but very high edge density, often containing large embedded cliques. It includes 50 instances specifically designed for graph coloring and 37 for clique problems.
    
\end{itemize}


\begin{table}[ht]
\centering
\caption{Statistics of each graph dataset. }
\label{tab:statistoc_dataset}
\scriptsize
\begin{tabular}{l|c|cc|cc|c}
\toprule
\multirow{2}{*}{Dataset} & \multirow{2}{*}{\#Graphs} & \multicolumn{2}{c|}{\(|V|\)} & \multicolumn{2}{c|}{\(|E|\)} & \multirow{2}{*}{Avg  Density} \\
& & Range & Avg & Range & Avg &  \\
\midrule
Network & 139 & [34,58790782] & 880272.3 & [4,106349209] & 3951924.5 & 0.0246 \\
SNAP & 51 & [1912,9912553] & 778605.7 & [17239,42851237] & 2389405.6 & 0.0012 \\
DIMACS10 & 134 & [34,67108864] & 4325165.9 & [78,907090182] & 33946494.4 & 0.0088 \\
DIMACS2 & 87 & [125,4000] & 652.1 & [1472,5506380] & 596556.0 &  0.5591 \\
Total & 411 & [34,67108864] & 1804612.4 & [4,907090182] & 12763549.5 & 0.1297 \\
\bottomrule
\end{tabular}
\end{table}

We show statistical information of each of these datasets in Table \ref{tab:statistoc_dataset}.
Together, these datasets cover both real-world graphs and early artificially generated graphs. In particular, the first three datasets consist mainly of massive, sparse graphs derived from real applications.
These graphs are widely used as benchmark graphs in the existing work for evaluating graph coloring, as well as other graph algorithms. 

\subsection{Overall Comparison}


\begin{table}[ht]
\centering
\caption{Summary of experimental results over all datasets. 
Each entry is reported as ``best (average hits)''.}
\label{tab:stats}
\resizebox{\textwidth}{!}{
\begin{tabular}{l|cc|cc|cc|cc}
\toprule
Dataset & \multicolumn{2}{c|}{RECOL} & \multicolumn{2}{c|}{FastColor} & \multicolumn{2}{c|}{GC-SLIM} & \multicolumn{2}{c}{LS-WGCP} \\
 & min & strict min & min & strict min & min & strict min & min & strict min \\
\midrule
Network & 134 (9.19) & 19 (4.89) & 110 (9.57) & 1 (2.00) & 93 (9.92) & 1 (5.00) & 93 (8.89) & 3 (3.00) \\
SNAP & 51 (9.59) & 2 (1.50) & 49 (9.76) & 0 & 28 (10.00) & 0 & 32 (8.09) & 0 \\
DIMACS10 & 111 (9.33) & 7 (4.43) & 107 (9.44) & 2 (10.00) & 45 (10.00) & 0 & 61 (8.85) & 9 (6.89) \\
DIMACS2 & 44 (7.93) & 11 (4.45) & 7 (6.00) & 0 & 63 (9.86) & 33 (9.73) & 28 (9.18) & 6 (7.83) \\
\midrule
Total & 340 (9.13) & 39 (4.51) & 273 (9.46) & 3 (7.33) & 263 (9.91) & 141 (9.84) & 214 (8.80) & 24 (7.12) \\
\bottomrule
\end{tabular}
}
\end{table}


Due to the randomized nature of these algorithms, each algorithm was executed ten times on every graph with a time limit of 60 seconds, using random seeds ranging from 0 to 9. 
For each graph, we record both the best coloring solution obtained across the ten runs and the average solution quality. 

Due to the large number of graphs (139 + 51 + 134 + 87 = 411 instances in total), the complete detailed results for all graphs are provided in the appended file.
In Table~\ref{tab:stats}, we report the aggregated statistics for each dataset. Here, \emph{min} denotes the number of graphs on which an algorithm’s solution (best and average over 10 runs) is no worse than the others, while \emph{strict min} indicates the number of graphs where it strictly outperforms all competitors. 
Each entry in the table follows the format ``best (average hits),'' where ``best'' refers to the number of graphs that the corresponding algorithm achieves the minimum and strictly minimum coloring number, and ``average hits'' indicates the average number of runs (out of 10) in which the algorithm achieves the minimum and strictly minimum coloring number.

As shown in Table~\ref{tab:stats}, \textsc{RECOL} clearly achieves smaller coloring numbers than the other algorithms in the first three sets of benchmark graphs. It also outperforms \emph{FastColor}, which is perhaps the best-known algorithm specifically designed for solving large graphs under short-term time limits. However, for the DIMACS2 set, \textsc{GC-SLIM}, the SAT-based long-term search algorithm still achieves the best results within 1 minute.



We also present detailed results for graphs with more than 200,000 vertices in Tables~\ref{tab:samples1}–\ref{tab:samples3}. Each table lists one graph per row, reporting the number of vertices ($|V|$) and edges ($|E|$), followed by the results for each algorithm. For each algorithm, the columns indicate the best solution (\emph{Min}), the average solution quality across 10 runs (\emph{Avg}), and the number of runs achieving the minimum (\#Hit). Unique minima are highlighted in bold.

As shown, for these large graphs, \textsc{RECOL} demonstrates even clearer advantages: it is worse than \textsc{FastColor} on only 2 graphs and worse than \textsc{LS-WGCP} on only 4 graphs, while achieving strictly better solutions than all other algorithms on 16 graphs. These results further confirm the effectiveness of our algorithm on very large graphs.

{

\renewcommand{\arraystretch}{0.8}
\scriptsize
\begin{xltabular}{\textwidth}{lrr|XXX|XXX|XXX|XXX}
\caption{Detailed results on Network Repository ($|V|\geq2\times 10^5$).}\\
\label{tab:samples1}\\
\toprule
Graph & $|V|$ & $|E|$ & \multicolumn{3}{c|}{RECOL} & \multicolumn{3}{c|}{FastColor} & \multicolumn{3}{c|}{LS-WGCP} & \multicolumn{3}{c}{GC-SLIM}  \\
 & $\times 10^5$& $\times 10^5$& Min & Avg & \#Hit & Min & Avg & \#Hit & Min & Avg & \#Hit & Min & Avg & \#Hit  \\
\midrule
sc-pwtk & 2 & 56 & 30 & \textbf{30.2} & 8 & 36 & 36.0 & 10 & 39 & 42.5 & 1 & \textbf{29} & 32.5 & 5 \\
ca-d...2010 & 2 & 7 & 75 & 75.0 & 10 & 75 & 75.0 & 10 & 75 & 75.0 & 10 & 75 & 75.0 & 10 \\
ca-c...seer & 2 & 8 & 87 & 87.0 & 10 & 87 & 87.0 & 10 & 87 & 87.0 & 10 & 87 & 87.0 & 10 \\
ca-d...2012 & 3 & 10 & 114 & 114.0 & 10 & 114 & 114.0 & 10 & 114 & 114.0 & 10 & 114 & 114.0 & 10 \\
ca-M...iNet & 3 & 8 & 25 & 25.0 & 10 & 25 & 25.0 & 10 & 25 & 25.0 & 10 & 25 & 25.0 & 10 \\
soc-...lows & 4 & 7 & 6 & 6.0 & 10 & 6 & 6.0 & 10 & 6 & 6.3 & 7 & 8 & 8.0 & 10 \\
sc-msdoor & 4 & 93 & \textbf{31} & \textbf{31.0} & 10 & 35 & 35.0 & 10 & 38 & 40.2 & 1 & 35 & 35.0 & 6 \\
soc-...tube & 4 & 19 & \textbf{22} & \textbf{22.8} & 2 & 25 & 25.0 & 10 & 26 & 26.4 & 6 & - & - & - \\
web-...2004 & 5 & 71 & 432 & 432.0 & 10 & 432 & 432.0 & 10 & 432 & 432.0 & 10 & - & - & - \\
soc-flickr & 5 & 31 & \textbf{93} & \textbf{93.7} & 3 & 101 & 101.0 & 10 & 111 & 113.9 & 1 & - & - & - \\
soc-...ious & 5 & 13 & 21 & 21.0 & 10 & 21 & 21.0 & 10 & 21 & 21.0 & 10 & - & - & - \\
ca-c...dblp & 5 & 152 & 337 & 337.0 & 10 & 337 & 337.0 & 10 & 337 & 337.0 & 10 & - & - & - \\
soc-...uare & 6 & 32 & 31 & 31.0 & 10 & 31 & 31.0 & 10 & 35 & 36.6 & 1 & - & - & - \\
web-...Stan & 6 & 66 & 29 & 29.0 & 10 & 29 & 29.0 & 10 & 29 & 29.0 & 10 & 29 & 29.0 & 10 \\
soc-digg & 7 & 59 & \textbf{54} & \textbf{54.9} & 1 & 63 & 63.0 & 10 & 60 & 62.1 & 1 & - & - & - \\
sc-ldoor & 9 & 207 & \textbf{31} & \textbf{31.0} & 10 & 35 & 35.0 & 10 & 40 & 40.7 & 5 & - & - & - \\
ca-h...2009 & 10 & 563 & 2209 & 2209 & 10 & 2209 & 2209 & 2 & 2209 & 2209 & 10 & - & - & - \\
inf-...t-PA & 10 & 15 & 4 & 4.0 & 10 & 4 & 4.0 & 10 & 4 & 4.0 & 10 & - & - & - \\
rt-r...rawl & 11 & 22 & 13 & 13.0 & 10 & 13 & 13.0 & 10 & 13 & 13.0 & 10 & - & - & - \\
scc\_...rawl & 11 & 0 & 20 & 20.0 & 10 & 20 & 20.0 & 10 & 20 & 20.0 & 10 & - & - & - \\
soc-...snap & 11 & 29 & \textbf{23} & \textbf{23.0} & 10 & 25 & 25.0 & 10 & 28 & 28.2 & 8 & - & - & - \\
soc-lastfm & 11 & 45 & \textbf{19} & \textbf{19.6} & 4 & 21 & 21.0 & 10 & 22 & 22.1 & 9 & - & - & - \\
soc-pokec & 16 & 223 & 29 & 29.0 & 10 & 29 & 29.0 & 10 & 29 & 30.5 & 2 & - & - & - \\
tech...tter & 16 & 110 & 67 & \textbf{67.0} & 10 & 68 & 68.0 & 10 & 67 & 68.7 & 1 & - & - & - \\
web-...2009 & 18 & 45 & 31 & 31.0 & 10 & 31 & 31.0 & 10 & 31 & 31.0 & 10 & - & - & - \\
inf-...t-CA & 19 & 27 & 4 & 4.0 & 10 & 4 & 4.0 & 10 & 4 & 4.1 & 9 & - & - & - \\
soc-...ster & 25 & 79 & \textbf{34} & \textbf{34.7} & 3 & 36 & 36.0 & 10 & 36 & 36.7 & 5 & - & - & - \\
socf...anon & 29 & 209 & \textbf{24} & \textbf{24.5} & 5 & 25 & 26.8 & 2 & 36 & 38.9 & 1 & - & - & - \\
soc-orkut & 29 & 1063 & \textbf{66} & \textbf{66.5} & 1 & - & - & - & 101 & 101.4 & 6 & - & - & - \\
socf...anon & 30 & 236 & \textbf{26} & \textbf{26.3} & 7 & 28 & 28.0 & 10 & 43 & 43.9 & 3 & - & - & - \\
soc-...rnal & 40 & 279 & 214 & 214.0 & 10 & 214 & 214.0 & 10 & 214 & 214.0 & 10 & - & - & - \\
inf-...-usa & 239 & 288 & 4 & 4.0 & 10 & 4 & 4.0 & 10 & - & - & - & - & - & - \\
socf...-uni & 587 & 922 & \textbf{6} & \textbf{6.0} & 10 & 7 & 7.7 & 3 & - & - & - & - & - & - \\
\bottomrule
\end{xltabular}
}
{
\scriptsize
\renewcommand{\arraystretch}{0.8}
\begin{xltabular}{\textwidth}{lrr|XXX|XXX|XXX|XXX}
\caption{Detailed results on SNAP ($|V|\geq2\times 10^5$).}\\
\label{tab:samples2}\\
\toprule
Graph & $|V|$ & $|E|$ & \multicolumn{3}{c|}{RECOL} & \multicolumn{3}{c|}{Fastcolor} & \multicolumn{3}{c|}{LS-WGCP} & \multicolumn{3}{c}{GC-SLIM}  \\
 & $\times 10^5$& $\times 10^5$& Min & Avg & \#Hit & Min & Avg & \#Hit & Min & Avg & \#Hit & Min & Avg & \#Hit  \\
\midrule
Amazon0302 & 2 & 8 & 7 & 7.0 & 10 & 7 & 7.0 & 10 & 7 & 7.0 & 10 & 7 & 7.0 & 10 \\
Emai...uAll & 2 & 3 & 18 & \textbf{18.1} & 9 & 18 & 18.5 & 5 & 23 & 23.2 & 8 & 21 & 21.8 & 2 \\
web-...ford & 2 & 19 & 61 & 61.0 & 10 & 61 & 61.0 & 10 & 61 & 61.0 & 10 & 64 & 64.0 & 10 \\
web-...Dame & 3 & 10 & 155 & 155.0 & 10 & 155 & 155.0 & 10 & 155 & 155.0 & 10 & 155 & 155.0 & 10 \\
Amazon0312 & 4 & 23 & 11 & 11.0 & 10 & 11 & 11.0 & 10 & 11 & 11.1 & 9 & 12 & 12.0 & 10 \\
Amazon0601 & 4 & 24 & 11 & 11.0 & 10 & 11 & 11.0 & 10 & 11 & 11.1 & 9 & 12 & 12.0 & 10 \\
Amazon0505 & 4 & 24 & 11 & 11.0 & 10 & 11 & 11.0 & 10 & 11 & 11.4 & 6 & 12 & 12.0 & 10 \\
web-...Stan & 6 & 66 & 201 & 201.0 & 10 & 201 & 201.0 & 10 & 201 & 201.0 & 10 & - & - & - \\
web-Google & 9 & 43 & 44 & 44.0 & 10 & 44 & 44.0 & 10 & 44 & 44.0 & 10 & - & - & - \\
soc-...hips & 16 & 223 & 29 & 29.0 & 10 & 29 & 29.0 & 10 & 29 & 29.7 & 4 & - & - & - \\
WikiTalk & 23 & 46 & \textbf{48} & \textbf{48.8} & 2 & 49 & 49.0 & 10 & 62 & 63.5 & 2 & - & - & - \\
soc-...nal1 & 48 & 428 & 321 & 321.0 & 10 & 321 & 321.0 & 10 & 334 & 334.0 & 10 & - & - & - \\
cit-...ents & 60 & 165 & 11 & 11.0 & 10 & 11 & 11.0 & 10 & 16 & 17.5 & 2 & - & - & - \\
Cit-HepTh & 99 & 3 & 24 & 24.0 & 10 & 24 & 24.0 & 10 & 24 & 24.6 & 4 & - & - & - \\
Cit-HepPh & 99 & 4 & 19 & 19.0 & 10 & 19 & 19.0 & 10 & 20 & 22.1 & 1 & - & - & - \\
\bottomrule
\end{xltabular}

\begin{xltabular}{\textwidth}{lrr|XXX|XXX|XXX|XXX}
\caption{Detailed results on Dimacs10 ($|V|\geq2\times 10^5$).}\\
\label{tab:samples3}\\
\toprule
Graph & $|V|$ & $|E|$ & \multicolumn{3}{c|}{RECOL} & \multicolumn{3}{c|}{Fastcolor} & \multicolumn{3}{c|}{LS-WGCP} & \multicolumn{3}{c}{GC-SLIM}  \\
 & $\times 10^5$& $\times 10^5$& Min & Avg & \#Hit & Min & Avg & \#Hit & Min & Avg & \#Hit & Min & Avg & \#Hit  \\
\midrule
coAu...seer & 2 & 8 & 87 & 87.0 & 10 & 87 & 87.0 & 10 & 87 & 87.0 & 10 & 87 & 87.0 & 10 \\
kron...gn18 & 2 & 211 & \textbf{36} & \textbf{36.6} & 4 & 37 & 37.3 & 7 & 54 & 55.4 & 3 & 38 & 38.0 & 10 \\
rgg\_...8\_s0 & 2 & 28 & 15 & 15.0 & 10 & 15 & 15.0 & 10 & 15 & 15.0 & 10 & 15 & 15.0 & 10 \\
kron...gn18 & 2 & 211 & \textbf{36} & \textbf{36.6} & 4 & 37 & 37.1 & 9 & 57 & 58.5 & 3 & 38 & 38.0 & 10 \\
dela...\_n18 & 2 & 7 & 5 & 5.0 & 10 & 5 & 5.0 & 10 & \textbf{4} & \textbf{4.0} & 10 & 5 & 5.0 & 10 \\
cita...seer & 2 & 11 & 13 & 13.0 & 10 & 13 & 13.0 & 10 & 13 & 13.0 & 10 & 13 & 13.0 & 10 \\
coAu...DBLP & 2 & 9 & 115 & 115.0 & 10 & 115 & 115.0 & 10 & 115 & 115.0 & 10 & 115 & 115.0 & 10 \\
cnr-2000 & 3 & 27 & 84 & 84.0 & 10 & 84 & 84.0 & 10 & 84 & 84.0 & 10 & 84 & 84.0 & 10 \\
coPa...seer & 4 & 160 & 845 & 845.0 & 10 & 845 & 845.0 & 10 & 845 & 845.0 & 10 & - & - & - \\
af\_shell9 & 5 & 85 & 20 & 20.0 & 10 & 20 & 20.0 & 10 & 24 & 25.0 & 2 & - & - & - \\
kron...gn19 & 5 & 435 & 38 & \textbf{38.6} & 4 & 38 & 40.8 & 1 & 65 & 65.3 & 7 & - & - & - \\
rgg\_...9\_s0 & 5 & 65 & 8 & 8.0 & 10 & 8 & 8.0 & 10 & \textbf{7} & \textbf{7.5} & 5 & - & - & - \\
dela...\_n19 & 5 & 15 & 5 & 5.0 & 10 & 5 & 5.0 & 10 & \textbf{4} & \textbf{4.4} & 6 & - & - & - \\
kron...gn19 & 5 & 435 & \textbf{38} & \textbf{38.3} & 7 & 39 & 43.2 & 3 & 65 & 65.0 & 10 & - & - & - \\
coPa...DBLP & 5 & 152 & 337 & 337.0 & 10 & 337 & 337.0 & 10 & 337 & 337.0 & 10 & - & - & - \\
eu-2005 & 8 & 161 & 387 & 387.0 & 10 & 387 & 387.0 & 10 & 387 & 387.0 & 10 & - & - & - \\
audikw1 & 9 & 383 & \textbf{41} & \textbf{41.7} & 3 & 42 & 42.0 & 10 & 51 & 51.5 & 5 & - & - & - \\
ldoor & 9 & 227 & \textbf{32} & \textbf{32.0} & 10 & 33 & 33.0 & 10 & 42 & 43.0 & 1 & - & - & - \\
ecology2 & 9 & 19 & 2 & 2.0 & 10 & 2 & 2.0 & 10 & 4 & 4.0 & 10 & - & - & - \\
ecology1 & 10 & 19 & 2 & 2.0 & 10 & 2 & 2.0 & 10 & 4 & 4.0 & 10 & - & - & - \\
NACA0015 & 10 & 31 & 5 & 5.0 & 10 & 5 & 5.0 & 10 & 6 & 6.0 & 10 & - & - & - \\
dela...\_n20 & 10 & 31 & 5 & 5.0 & 10 & 5 & 5.0 & 10 & 5 & 5.0 & 10 & - & - & - \\
er-f...le20 & 10 & 109 & 8 & 8.0 & 10 & 8 & 8.0 & 10 & 10 & 10.4 & 6 & - & - & - \\
rgg\_...0\_s0 & 10 & 104 & 17 & 17.0 & 10 & 17 & 17.0 & 10 & 17 & 17.0 & 10 & - & - & - \\
kron...gn20 & 10 & 892 & - & - & - & \textbf{46} & \textbf{46.0} & 10 & - & - & - & - & - & - \\
thermal2 & 12 & 36 & 5 & 5.0 & 10 & 5 & 5.0 & 10 & 6 & 6.0 & 10 & - & - & - \\
in-2004 & 13 & 139 & 489 & 489.0 & 10 & 489 & 489.0 & 10 & 489 & 489.0 & 10 & - & - & - \\
belg....osm & 14 & 15 & 3 & 3.0 & 10 & 3 & 3.0 & 10 & 3 & 3.0 & 10 & - & - & - \\
af\_shell10 & 15 & 255 & 15 & 15.0 & 10 & 15 & 15.0 & 10 & 26 & 27.3 & 1 & - & - & - \\
G3\_circuit & 15 & 30 & 3 & 3.0 & 10 & 3 & 3.0 & 10 & 3 & 3.6 & 4 & - & - & - \\
kkt\_power & 20 & 64 & 11 & 11.0 & 10 & 11 & 11.0 & 10 & 11 & 11.0 & 10 & - & - & - \\
rgg\_...1\_s0 & 20 & 245 & 18 & 18.0 & 10 & 18 & 18.0 & 10 & 18 & 18.0 & 10 & - & - & - \\
er-f...le21 & 20 & 229 & 9 & 9.0 & 10 & 9 & 9.0 & 10 & 11 & 11.0 & 10 & - & - & - \\
dela...\_n21 & 20 & 62 & 5 & 5.0 & 10 & 5 & 5.0 & 10 & 5 & 6.2 & 4 & - & - & - \\
pack...b050 & 21 & 308 & \textbf{10} & \textbf{10.9} & 1 & 11 & 11.0 & 10 & 14 & 14.7 & 3 & - & - & - \\
neth....osm & 22 & 24 & 4 & 4.0 & 10 & 4 & 4.0 & 10 & 4 & 4.0 & 10 & - & - & - \\
M6 & 35 & 105 & 5 & 5.0 & 10 & 5 & 5.0 & 10 & 6 & 6.3 & 7 & - & - & - \\
nlpkkt120 & 35 & 466 & 2 & 2.0 & 10 & 2 & 2.0 & 10 & 8 & 8.1 & 9 & - & - & - \\
333SP & 37 & 111 & 5 & 5.0 & 10 & 5 & 5.0 & 10 & 7 & 7.0 & 10 & - & - & - \\
AS365 & 37 & 113 & 5 & 5.0 & 10 & 5 & 5.0 & 10 & 7 & 7.0 & 10 & - & - & - \\
vent...vel3 & 40 & 80 & 4 & 4.0 & 10 & 4 & 4.0 & 10 & 4 & 4.0 & 10 & - & - & - \\
NLR & 41 & 124 & 5 & 5.0 & 10 & 5 & 5.0 & 10 & 8 & 8.0 & 10 & - & - & - \\
er-f...le22 & 41 & 479 & 9 & 9.0 & 10 & 9 & 9.0 & 10 & 13 & 13.7 & 3 & - & - & - \\
dela...\_n22 & 41 & 125 & 5 & 5.0 & 10 & 5 & 5.0 & 10 & 8 & 8.0 & 10 & - & - & - \\
rgg\_...2\_s0 & 41 & 605 & 11 & 11.0 & 10 & 11 & 11.0 & 10 & 16 & 16.1 & 9 & - & - & - \\
huge...0000 & 45 & 68 & 3 & 3.0 & 10 & 3 & 3.0 & 10 & 4 & 4.0 & 10 & - & - & - \\
chan...b050 & 48 & 426 & 4 & 4.0 & 10 & 4 & 4.0 & 10 & 13 & 13.1 & 9 & - & - & - \\
cage15 & 51 & 470 & 12 & \textbf{12.6} & 4 & 12 & 12.7 & 3 & 18 & 18.0 & 10 & - & - & - \\
huge...0000 & 58 & 87 & 3 & 3.0 & 10 & 3 & 3.0 & 10 & 4 & 4.0 & 8 & - & - & - \\
huge...0010 & 65 & 98 & 3 & 3.0 & 10 & 3 & 3.0 & 10 & 4 & 4.0 & 5 & - & - & - \\
italy.osm & 66 & 70 & 3 & 3.0 & 10 & 3 & 3.0 & 10 & 3 & 4.0 & 2 & - & - & - \\
adaptive & 68 & 136 & 2 & 2.0 & 10 & 2 & 2.0 & 10 & - & - & - & - & - & - \\
huge...0020 & 71 & 106 & 3 & 3.0 & 10 & 3 & 3.0 & 10 & - & - & - & - & - & - \\
grea....osm & 77 & 81 & 3 & 3.0 & 10 & 3 & 3.0 & 10 & 4 & 4.9 & 1 & - & - & - \\
nlpkkt160 & 83 & 1105 & 2 & 2.0 & 10 & 2 & 2.0 & 10 & - & - & - & - & - & - \\
rgg\_...3\_s0 & 83 & 1119 & 19 & 19.0 & 10 & 19 & 19.0 & 10 & - & - & - & - & - & - \\
dela...\_n23 & 83 & 251 & 5 & 5.0 & 10 & 5 & 5.0 & 10 & - & - & - & - & - & - \\
er-f...le23 & 83 & 1003 & - & - & - & \textbf{11} & \textbf{11.0} & 10 & - & - & - & - & - & - \\
germ....osm & 115 & 123 & 4 & 4.0 & 10 & 4 & 4.0 & 10 & 5 & 5.0 & 7 & - & - & - \\
asia.osm & 119 & 127 & 4 & 4.0 & 10 & 4 & 4.0 & 10 & 5 & 5.0 & 10 & - & - & - \\
huge...0010 & 120 & 180 & 3 & 3.0 & 10 & 3 & 3.0 & 10 & - & - & - & - & - & - \\
road...tral & 140 & 169 & 4 & 4.0 & 10 & 4 & 4.0 & 10 & - & - & - & - & - & - \\
huge...0020 & 160 & 239 & 3 & 3.0 & 10 & 3 & 3.0 & 10 & - & - & - & - & - & - \\
nlpkkt200 & 162 & 2159 & 2 & 2.0 & 10 & 2 & 2.0 & 10 & - & - & - & - & - & - \\
dela...\_n24 & 167 & 503 & 5 & 5.0 & 10 & 5 & 5.0 & 10 & - & - & - & - & - & - \\
rgg\_...4\_s0 & 167 & 1802 & 21 & 21.0 & 10 & 21 & 21.0 & 10 & - & - & - & - & - & - \\
huge...0000 & 183 & 274 & 3 & 3.0 & 10 & 3 & 3.0 & 10 & - & - & - & - & - & - \\
uk-2002 & 185 & 2620 & 944 & 944.0 & 10 & 944 & 944.0 & 10 & - & - & - & - & - & - \\
huge...0010 & 194 & 291 & 3 & 3.0 & 10 & 3 & 3.0 & 10 & - & - & - & - & - & - \\
huge...0020 & 211 & 317 & 3 & 3.0 & 10 & 3 & 3.0 & 10 & - & - & - & - & - & - \\
road\_usa & 239 & 288 & 4 & 4.0 & 10 & 4 & 4.0 & 10 & - & - & - & - & - & - \\
nlpkkt240 & 279 & 3732 & 2 & 2.0 & 10 & 2 & 2.0 & 10 & - & - & - & - & - & - \\
europe.osm & 509 & 540 & 4 & 4.0 & 10 & 4 & 4.0 & 10 & - & - & - & - & - & - \\
\bottomrule
\end{xltabular}
}


To provide a more intuitive comparison among the algorithms, we plot line charts for each dataset in Figure~\ref{fig:placeholder}. The x-axis represents time, and the y-axis shows the number of graphs for which an algorithm achieves the best solution (among all algorithms) within 60 seconds. After investigating the code of \textsc{GC-SLIM}, we found that it does not report the time to reach each solution, so its performance is indicated by a horizontal dashed line. 
In sum, these plots illustrate the speed at which each algorithm attains optimal or near-optimal solutions across the datasets. Clearly, RECOL consistently solves more graphs than the other algorithms in 1 minute.
Although in the first 10 seconds, RECOL solves strictly no more graphs than others, indicating its good near-real-time ability.
Overall, these results demonstrate that \textsc{RECOL} is highly effective for short-time graph coloring.

\begin{figure}[t]
    \centering
    \includegraphics[width=1\linewidth]{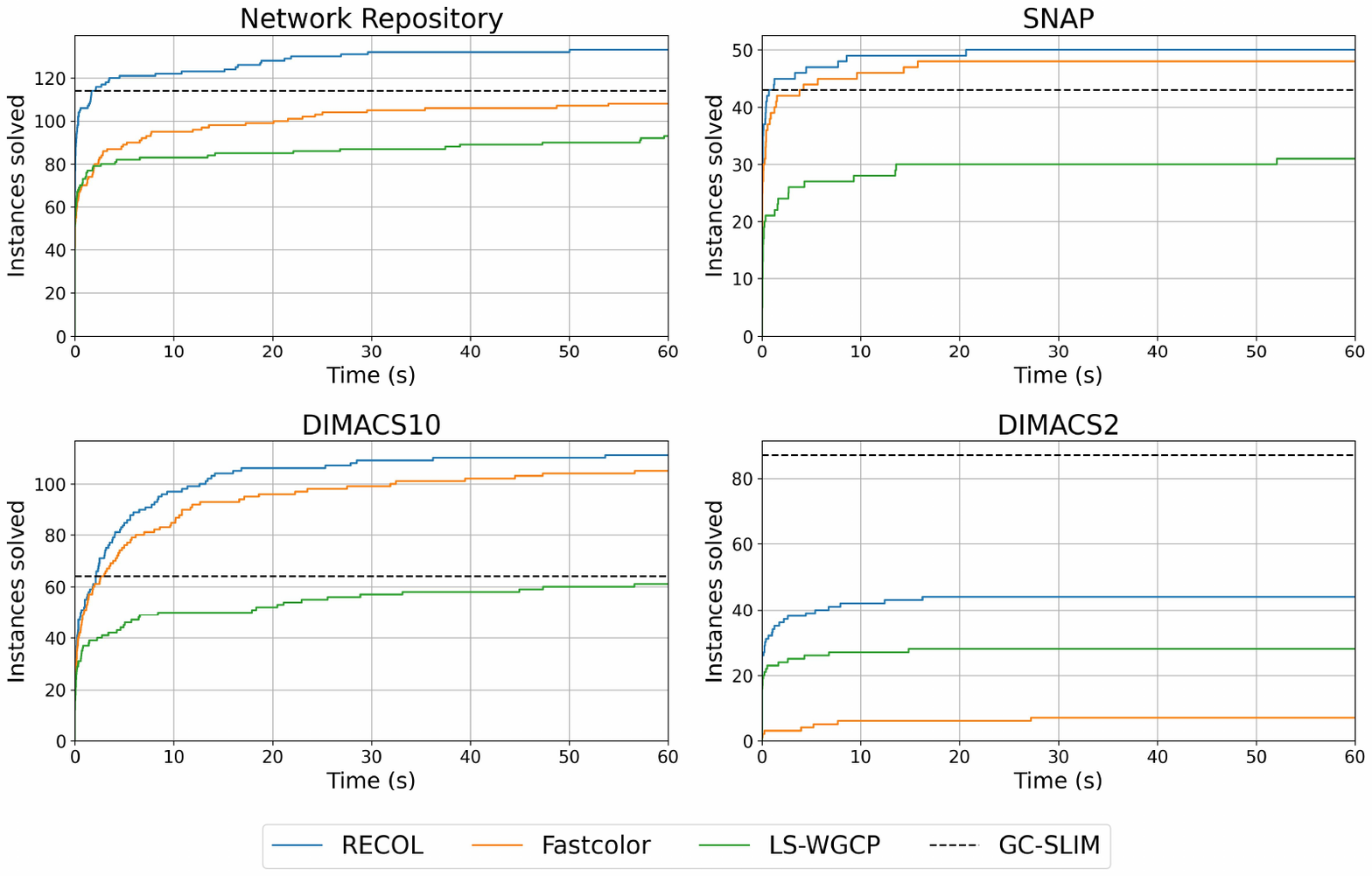}
    \caption{Cumulative number of graphs for which each algorithm attains the best solution over time.}
    \label{fig:placeholder}
\end{figure}


\subsection{Effectiveness of Reductions}



We conducted a further empirical study to evaluate the effectiveness of our reduction rules. 
Specifically, we compared four variants: the full \textsc{RECOL}, \textsc{RECOL} without degree reduction (\emph{No-Degree}), \textsc{RECOL} without domination reduction (\emph{No-Dominate}), and \textsc{RECOL} without independent set reduction and complement crown reduction (\emph{No-Ind-Crown}). 
Note that the independent set reduction (Lemma~\ref{independentreduction}) also covers the case of complement crown reduction when $|C|=1$ (Lemma~\ref{crown1reduction}), since when $|\overline{N}(u)| \leq 1$, the set $\overline{N}[u]$ forms an independent set as well.)
For this reason, we consider these two reduction rules together.

All experiments were conducted on the four benchmark collections under the same settings as before: a time limit of 60 seconds, ten runs with random seeds 0–9. 
Similar to Table~\ref{tab:stats}, we present statistical results in Tables~\ref{tab:reduction_node} and~\ref{tab:reduction_result}, which report the effect of the reductions on the number of remaining vertices and the coloring number, respectively. Here, min and strict min follow the same notation as in Table~\ref{tab:stats}.

\begin{table}[ht]
\centering
\scriptsize
\caption{Ablation results on reduction size over all datasets. Each entry shows the number of instances with the minimum remaining vertices.}
\label{tab:reduction_node}
\scriptsize
\begin{tabular}{l|cc|cc|cc|cc}
\toprule
Dataset & \multicolumn{2}{c|}{RECOL} & \multicolumn{2}{c|}{No-Ind-Crown} & \multicolumn{2}{c|}{No-Dominate} & \multicolumn{2}{c}{No-Degree} \\
 & min & strict min & min & strict min & min & strict min & min & strict min \\
\midrule
NETWORK  & 139 & 3 & 128 & 0 & 133 & 0 & 11 & 0 \\
SNAP     & 51 & 0 & 49 & 0 & 49 & 0 & 3 & 0 \\
DIMACS10 & 125 & 0 & 123 & 0 & 124 & 0 & 33 & 0 \\
DIMACS2  & 69 & 0 & 69 & 0 & 69 & 0 & 69 & 0 \\
\midrule
Total    & 384 & 3 & 369 & 0 & 375 & 0 & 116 & 0 \\
\bottomrule
\end{tabular}
\end{table}

Across the three datasets other than DIMACS2, \textsc{RECOL} with all reduction rules consistently shows an advantage over the ablated variants.
Without exception, the degree reduction removes the most vertices because the degrees of most vertices in these graphs are small.
Nevertheless, the domination reduction, independent set reduction, and complement crown reduction play an important role in many graphs. 
For example, on the \emph{tech-routers-rf} graph with 2113 vertices from the Network Repository, \textsc{RECOL} reduces the graph to an average of 463.9 vertices, compared to 465.4 with \emph{No-Ind-Crown}, 465.6 with \emph{No-Dominate}, and 2113.0 with \emph{No-Degree}.
On the dense DIMACS2 instances, however, none of the reduction rules can remove any vertices, resulting in identical outcomes across all variants.

\begin{table}[ht]
\centering
\scriptsize
\caption{Ablation results over all datasets. 
Each entry is reported as ``best (average hits)''.
}
\label{tab:reduction_result}
\begin{tabular}{l|cc|cc|cc|cc}
\toprule
Dataset & \multicolumn{2}{c|}{RECOL} & \multicolumn{2}{c|}{No-Ind-Crown} & \multicolumn{2}{c|}{No-Dominate} & \multicolumn{2}{c}{No-Degree} \\
 & min & strict min  & min & strict min  & min & strict min  & min & strict min \\
\midrule
NETWORK & 139 (9.18) & 0 & 139 (9.18) & 0 & 138 (9.23) & 0 & 130 (9.28) & 0 \\
SNAP & 51 (9.59) & 0 & 51 (9.59) & 0 & 51 (9.59) & 0 & 49 (9.43) & 0 \\
DIMACS10 & 126 (9.40) & 0 & 126 (9.40) & 0 & 126 (9.40) & 0 & 124 (9.37) & 1 (1) \\
DIMACS2 & 68 (6.91) & 1 (1) & 67 (7.04) & 0  & 65 (7.12) & 1 (1) & 64 (7.30) & 0  \\
\midrule
Total & 384 (8.77) & 1 (1) & 383 (8.80) & 0  & 380 (8.84) & 1 (1) & 367 (8.84) & 1 (1) \\
\bottomrule
\end{tabular}
\end{table}

Regarding the final coloring performance (Table~\ref{tab:reduction_result}), the differences among the variants are smaller; nevertheless, the benefits of the reductions remain evident. All the reduction rules contribute positively to the overall effectiveness, particularly Degree Reduction, and their combined application consistently delivers the strongest performance across datasets.

\section{Summary}

This paper proposed a novel reduction-rule-based heuristic algorithm for graph coloring on large-scale graphs, named \textit{RECOL}. We introduced four reduction rules: degree-based reduction, domination reduction, complement-crown decomposition, and independent set reduction. 
The algorithm alternates between computing the upper and lower bounds of the chromatic number and applying these reduction rules. Once the reduction process is completed or reaches a time limit, a heuristic coloring algorithm is applied to the remaining graph.

Extensive experiments on four benchmark datasets demonstrate that \textit{RECOL} is highly efficient and effective on large, sparse graphs, yielding colorings with fewer colors within short time limits while substantially reducing graph size. Our ablation study shows that all reduction rules contribute to shrinking the graph size, with degree reduction having the most prominent impact. The other reduction rules, though individually less dominant, still provide meaningful improvements in many cases. This highlights the importance of carefully determining how to apply each reduction to balance solution quality with computational efficiency. Designing more adaptive strategies to optimize the integration of reduction rules thus constitutes a promising direction for future research.
\bibliographystyle{elsarticle-num} 
\bibliography{color}



\appendix

\end{document}